%% file: main.tex
\documentclass[hidelinks,onefignum,onetabnum]{siamart220329}

\usepackage{cite}
\usepackage{color}
\usepackage{bm,epsfig,epsf,url,dsfont}
\usepackage{graphicx}
\usepackage{newtxmath}      
\usepackage{subcaption}

\usepackage{hyperref,xcolor}
\hypersetup{
  colorlinks=true,
  linkcolor=blue,
  linktoc=all,
  citecolor=blue,
  urlcolor=blue,
  bookmarksnumbered=true
}
\usepackage{cleveref}

\newcommand{\NN}{{\mathbb{N}}}

\newcommand{\eps}{\epsilon}

\DeclareMathOperator{\opt}{opt}

\newcommand{\Ind}{\vvmathbb 1}
\let\1\Ind

\newenvironment{proofof}[1]{\medskip\noindent{\bf Proof of #1.}}{\hfill$\Box$\vspace*{1mm}\medskip}

\include{commands}

\begin{document}

\title{Shortest Paths without a Map,\\but with an Entropic Regularizer\thanks{A preliminary version of this
paper was published in the Proceedings of the 63rd Annual IEEE Symposium on Foundations of Computer
Science.}}

\author{
S\'ebastien Bubeck\thanks{ML Foundations Group, Microsoft Research, Redmond, WA, United States.
{\tt sebastien.bubeck@gmail.com}.}
\and
Christian Coester\thanks{Department of Computer Science, University of Oxford, Oxford, United Kingdom.
{\tt christian.coester@cs.ox.ac.uk}. Research was done while at Tel Aviv University, 
supported by the Israel Academy of Sciences and Humanities \& Council for Higher 
Education Excellence Fellowship Program for International Postdoctoral Researchers.}
\and
Yuval Rabani\thanks{School of Computer Science \& Engr., The Hebrew University of Jerusalem,
Jerusalem, Israel. {\tt yrabani@cs.huji.ac.il}. Research supported in part by ISF grants 3565-21 
and 389-22, and by BSF grant 2018687.}
}

\maketitle

\begin{abstract}
In a 1989 paper titled ``shortest paths without a map'', Papadimitriou and Yannakakis introduced an online model of searching in a weighted layered graph for a target node, while attempting to minimize the total length of the path traversed by the searcher. This problem, later called layered graph traversal, is parametrized by the maximum cardinality $k$ of a layer of the input graph. It is an online setting for dynamic programming, and it is known to be a rather general and fundamental model of online computing, which includes as special cases other acclaimed models. The deterministic competitive ratio for this problem was soon discovered to be exponential in $k$, and it is now nearly resolved: it lies between $\Omega(2^k)$ and $O(k2^k)$. Regarding the randomized competitive ratio, in 1993 Ramesh proved, surprisingly, that this ratio has to be at least $\Omega(k^2 / \log^{1+\epsilon} k)$ (for any constant $\epsilon > 0$). In the same paper, Ramesh also gave an $O(k^{13})$-competitive randomized online algorithm. Between 1993 and the results obtained in this paper, no progress has been reported on the randomized competitive ratio of layered graph traversal. In this work we show how to apply the mirror descent framework on a carefully selected evolving metric space, and obtain an $O(k^2)$-competitive randomized online algorithm.
This matches asymptotically an improvement of the aforementioned lower bound~\cite{BCR22}, which we announced (among other results) after the initial publication of the results here.
\end{abstract}



%


\section{Introduction}

\paragraph*{\bf Our results}
In this paper we present a randomized $O(k^2)$-competitive online algorithm
for width $k$ layered graph traversal. The problem, whose history is discussed
in detail below, is an online version of the conventional shortest path problem
(see~\cite{Sch12}), and of Bellman's dynamic programming paradigm. It is
therefore a very general framework for decision-making under uncertainty of
the future, and it includes as special cases other celebrated models of online
computing, such as metrical task systems. The following 1989 words of 
Papadimitriou and Yannakakis~\cite{PY89} still resonate today: ``the techniques 
developed [for layered graph traversal] will add to the scarce rigorous methodological 
arsenal of Artificial Intelligence.'' Our new upper bound on the competitive
ratio matches asymptotically the $\Omega(k^2)$ lower bound in~\cite{BCR22}, which was announced
following the initial publication of the results here. Our upper bound improves substantially 
over the previously known $O(k^{13})$ upper bound given by Ramesh~\cite{Ram93}.

\paragraph*{\bf Problem definition}
In layered graph traversal, a searcher attempts to find a short path from a starting node $a$ 
to a target node $b$ in an undirected graph $G$ with non-negative edge weights 
$w:E(G)\rightarrow\R_{+}$. It is assumed that the vertices of $G$ are partitioned into layers 
such that edges exist only between vertices of consecutive layers. The searcher knows initially 
only an upper bound $k$ on the number of nodes in a layer (this is called the {\em width} of $G$).
The search begins at $a$, which is the only node in the first layer of the graph (indexed layer $0$). 
We may assume that all the nodes of $G$ are reachable from $a$. When the searcher first reaches 
a node in layer $i$, the next layer $i+1$, the edges between layer $i$ and layer $i+1$, and the weights 
of these edges, are revealed to the searcher.\footnote{So, in particular, when the search starts at $a$, 
the searcher knows all the nodes of layer $1$, all the edges connecting $a$ to a node in layer $1$, 
and all the weights of such edges.} At this point, the searcher has to move forward from the current position 
(a node in layer $i$) to a new position (a node in layer $i+1$). This can be done along a shortest path 
through the layers revealed so far.\footnote{Notice that this is not necessarily the shortest path in $G$ 
between the nodes, because such a path may have to go through future layers that have not been 
revealed so far. In fact, some nodes in layer $i+1$ might not even be reachable at this point.} 
The target node $b$ occupies the last layer of the graph; the number of layers is unknown
to the searcher. The target gets revealed when the preceding layer is reached. The goal of 
the searcher is to traverse a path of total weight as close as possible to the shortest path 
connecting $a$ and $b$.\footnote{Note that the traversed path is not required to be simple 
or level-monotone.} The searcher is said to be $C$-competitive for some $C=C(k)$ iff for every
width $k$ input $(G,w)$, the searcher's path weight is at most $C\cdot w_G(a,b)$, where
$w_G(a,b)$ denotes the distance under the weight function $w$ between $a$ and $b$. The
competitive ratio of the problem is the best $C$ for which there is a $C$-competitive searcher.

Note that any bipartite graph is layered, and any graph can be converted 
into a bipartite graph by subdividing edges. Thus, the layered structure of the input graph is intended
primarily to parametrize the way in which the graph is revealed over time to the online algorithm.
Also notice that if the width $k$ is not known, it can be estimated, for instance by doubling the guess 
each time it is refuted. Therefore, the assumptions made in the definition of the problem are not
overly restrictive.

\subsection{Motivation and past work}
The problem has its origins in a paper of Baeza-Yates et al.~\cite{BCR88}, and possibly
earlier in game theory~\cite{Gal80}. In~\cite{BCR88}, motivated by applications in robotics,
the special case of a graph consisting of $k$ disjoint paths is proposed, under the name 
{\em the lost cow problem}. The paper gives a deterministic $9$-competitive algorithm
for $k=2$, and more generally a deterministic $2\frac{k^k}{(k-1)^{k-1}}+1\approx 2ek+1$ 
competitive algorithm for arbitrary $k$. A year later, Papadimitriou and Yannakakis
introduced the problem of layered graph traversal~\cite{PY89}. Their paper shows that
for $k=2$ the upper bound of $9$ still holds, and that the results of~\cite{BCR88} are
optimal for disjoint paths. Unfortunately, the upper bound of $9$ in~\cite{PY89} is the
trivial consequence of the observation that for $k=2$, the general case reduces to the disjoint
paths case. This is not true for general $k$. Indeed, Fiat et al.~\cite{FFKRRV91} give a $2^{k-2}$
lower bound, and an $O(9^k)$ upper bound on the deterministic competitive ratio in the general 
case. The upper bound was improved by Ramesh~\cite{Ram93} to $O(k^3 2^k)$, and
further by Burley~\cite{Bur96} to $O(k 2^k)$. Thus, currently the deterministic case is nearly 
resolved, asymptotically: the competitive ratio lies in $\Omega(2^k)\cap O(k2^k)$.

Investigation of the randomized competitive ratio was initiated in the aforementioned~\cite{FFKRRV91}
that gives a $\frac{k}{2}$ lower bound for the general case, and asymptotically tight $\Theta(\log k)$
upper and lower bounds for the disjoint paths case. In the randomized case, the searcher's strategy
is a distribution over moves, and it is $C$-competitive iff for every width $k$ input $(G,w)$, the expected 
weight of the searcher's path is at most $C\cdot w_G(a,b)$.
It is a ubiquitous phenomenon of online computing that randomization improves the competitive 
ratio immensely, often guaranteeing exponential asymptotic improvement (as happens in the
disjoint paths case of layered graph traversal). To understand why this might happen, one can
view the problem as a game between the designer of $G$ and the searcher in $G$. The game
alternates between the designer introducing a new layer
and the searcher moving to a node in the new layer. 
The designer is oblivious 
to the searcher's moves. Randomization obscures the predictability of the searcher's moves, and 
thus weakens the power of the designer.\footnote{This can be formalized through an appropriate definition
of the designer's information sets.} Following the results in~\cite{FFKRRV91} and
the recurring phenomenon of exponential improvement, a natural conjecture would have been that
the randomized competitive ratio of layered graph traversal is $\Theta(k)$. However, this natural
conjecture was rather quickly refuted in the aforementioned~\cite{Ram93} that surprisingly improves
the lower bound to $\Omega(k^2/\log^{1+\eps} k)$, which holds for all $\eps > 0$. Moreover, the
same paper also gives an upper bound of $O(k^{13})$. Thus, even though for the general case of
layered graph traversal the randomized competitive ratio cannot be logarithmic in the deterministic
competitive ratio, it is polylogarithmic in that ratio. The results of~\cite{Ram93} on randomized layered
graph traversal have not since been improved prior to our current paper. Very recently and following the
initial publication of the results in this paper, it was announced that the lower bound of~\cite{Ram93} 
was improved slightly to $\Omega(k^2)$~\cite{BCR22}.

\paragraph*{\bf Some Applications} Computing the optimal offline solution, a shortest path from source to target in a weighted layered
graph, is a simple example and also a generic framework for dynamic programming~\cite{Bel57}.
The online version has applications to the design and analysis of hybrid algorithms. In particular, the
disjoint paths case has applications in derandomizing online algorithms~\cite{FRRS94}, in the
design of divide-and-conquer online algorithms~\cite{FRR90,ABM93}, and in the design of advice
and learning augmented online algorithms~\cite{LV21,ACEPS20,BCKPV22}. 
In this context, Kao et al.~\cite{KRT93} resolve exactly the randomized competitive ratio of width $2$ 
layered graph traversal: it is roughly $4.59112$, precisely the solution for $x$ in the equation 
$\ln(x-1) = \frac{x}{x-1}$; see also~\cite{CL91}. For more in this vein, see also Kao et al.~\cite{KMSY94}. 

\paragraph{\bf Connections with Other Problems} 
Moreover, layered graph traversal is a very general model of online computing. For example, many online 
problems can be represented as chasing finite subsets of points in a metric space. In this problem, introduced 
by Chrobak and Larmore~\cite{CL91,CL93} under the name {\em metrical service systems}, the online
algorithm is given a sequence of bounded-cardinality subsets of points in a metric space. The algorithm
controls a server, initially at some point of the metric space. In response to a request, the algorithm must
move the server to one of the points of the requested subset, paying a cost equal to the distance traveled.
This problem is known to be {\bf equivalent} to layered graph traversal~\cite{FFKRRV91}. 
The width $k$ of the layered graph instance corresponds to the 
maximum cardinality of any request of the metrical service systems instance. (See Section~\ref{sec: applications}.)
Metrical service systems are a 
special case of metrical task systems, introduced by Borodin et al.~\cite{BLS87}. Width $k$ layered graph 
traversal includes as a special case metrical task systems in $k$-point metric spaces.\footnote{Indeed, this 
implies that while metrical service systems on $k$-point metrics are a special case of metrical task systems 
on $k$-point metrics, also metrical task systems on $k$-point metrics are a special case of metrical service 
systems using $k$-point requests.}
There is a tight bound of $2k-1$ on the deterministic competitive ratio of metrical task systems in any
$k$-point metric~\cite{BLS87}, and the randomized competitive ratio lies between $\Omega(\log k)$ and $O(\log^2 k)$ in all metric spaces \cite{BCLL19,BCR22}.
Thus, width $k$ layered graph traversal is strictly a more general problem than $k$-point metrical task systems. 
Another closely related problem is the $k$-taxi problem, whose best known lower bound for deterministic algorithms 
is obtained via a reduction from layered graph traversal~\cite{CK19}.

\subsection{Our methods}
Our techniques are based on the method of online mirror descent with entropic regularization
that was pioneered by Bubeck et al.~\cite{BCLLM18,BCLL19} in the context of the $k$-server
problem and metrical task systems, and further explored in this context in a number of recent 
papers~\cite{CL19,EL21,BC21}.
It is known that layered graph traversal is equivalent to its special case where the input graph is a tree~\cite{FFKRRV91}. Based on this reduction, we reduce width $k$ layered graph traversal to a problem that we name
the (depth $k$) \emph{evolving tree game}.

\paragraph{\bf Evolving Tree Game} In this game (defined formally in Section~\ref{sec: evolving tree}), one player, representing the algorithm,
must occupy a leaf in an evolving edge-weighted tree. Intuitively, the leaves correspond to the nodes in the current layer of the layered graph traversal instance. Another player, the adversary, is allowed to change the metric 
and topology of the tree using the following repertoire of operations: ($i$) continuously increase the 
weight of an edge incident to a leaf (see Fig.~\ref{fig:growth}); ($ii$) delete a leaf and the incident edge, and smooth 
the tree at the parent if its degree is now $2$  (see Fig.~\ref{fig:delete});\footnote{Smoothing is the reverse operation
of subdividing. In other words, smoothing is merging the two edges incident to a degree 
$2$ node. We maintain w.l.o.g. the invariant that the tree
has no degree $2$ node.} ($iii$) create two (or more) new leaves and 
connect them with weight $0$ edges to an existing leaf whose combinatorial depth is 
strictly smaller than $k$  (see Fig.~\ref{fig:fork}). The algorithm may move from leaf to leaf at any time, incurring \emph{movement cost} equal to the weight of the path between the leaves. If the algorithm occupies a leaf at 
the endpoint of a growing weight edge, it pays the increase in weight. We call this
the {\em service cost} of the algorithm. If the algorithm occupies a leaf that is being deleted, 
it must move to a different leaf prior to the execution of the topology change. At the end 
of the game, the total (movement + service) cost of the algorithm is compared against the adversary's cost, which is 
the weight of the lightest root-to-leaf path.\footnote{In fact, our algorithm can handle also the 
operation of reducing the weight of an edge, under the assumption that this operation incurs 
on both players a cost equal to the reduction in weight, if performed at their location.}

Layered graph traversal and the evolving tree game are connected by the following reduction, which we prove in Section~\ref{sec: applications}.
\begin{theorem}\label{thm:reductionShort}
If there is a $C$-competitive strategy for the evolving tree game on \emph{binary} trees of depth at most $k$, then there is a $C$-competitive strategy for traversing layered graphs of width $k$.
\end{theorem}

\paragraph{\bf Algorithm for the Evolving Tree Game} Mirror descent is used to generate a fractional online solution to the evolving tree game. 
The algorithm maintains a probability distribution on the leaves. A fractional solution can 
be converted easily on-the-fly into a randomized algorithm. As in~\cite{BCLLM18,BCLL19}
the analysis of our fractional algorithm for the evolving tree game is based on 
a potential function that combines (in our case, a modification of) the Bregman divergence 
associated with the entropic regularizer with a weighted depth potential. The Bregman 
divergence is used to bound the algorithm's {\em service cost} against the adversary's cost. The 
weighted depth potential is used to bound the algorithm's {\em movement cost} against 
its own service cost.

\begin{figure}
	\centering
	\begin{subfigure}{\linewidth}
		\centering
		\includegraphics[scale=0.3]{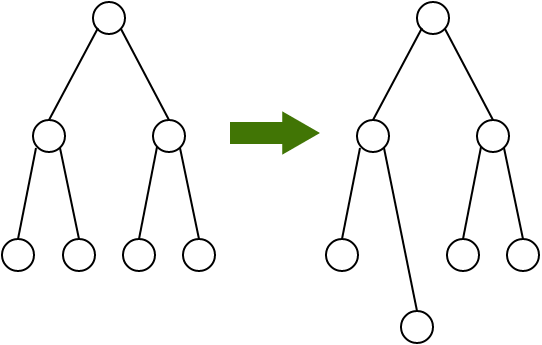}
		\caption{Continuous step.}\label{fig:growth}
	\end{subfigure}

	\vspace{0.7cm}
	
	\begin{subfigure}{\linewidth}
		\centering
		\includegraphics[scale=0.3]{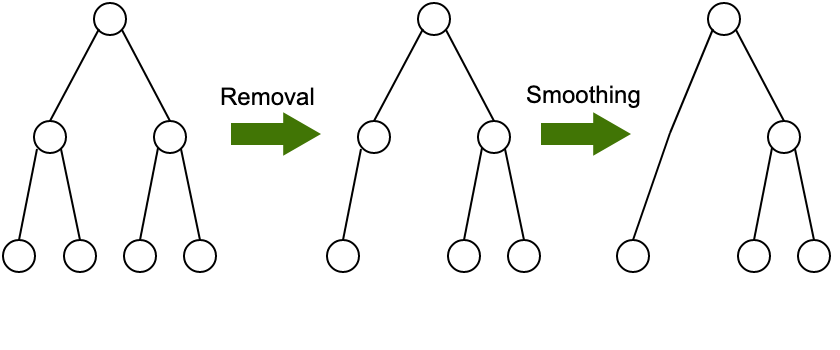}
		\caption{Delete step.}\label{fig:delete}
	\end{subfigure}

	\vspace{0.7cm}

	\begin{subfigure}{\linewidth}
		\centering
		\includegraphics[scale=0.3]{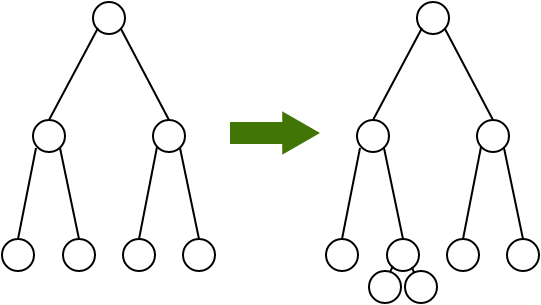}
		\caption{Fork step.}\label{fig:fork}
	\end{subfigure}

	\caption{Three types of metric and topology changes of the tree.}\label{fig:steps}
\end{figure}

However, in our setting, in contrast to~\cite{BCLLM18,BCLL19}, the metric on 
the set of leaves, and even the topology of the underlying tree, change dynamically. This 
poses a few new challenges to the approach. In particular, the potential function that works 
for metrical task systems is not invariant under the topology changes that are needed here. 
We resolve this problem by working with revised edge weights that slightly over-estimate the 
true edge weights. When a topology change would lead to an increase of the potential function 
(by reducing the combinatorial depth of some vertices), we prevent such an increase by 
downscaling the affected revised edge weights appropriately.

Even so, the extra cost incurred by the perturbation 
of entropy, which is required to handle distributions close to the boundary, cannot be handled 
in the same manner as in \cite{BCLLM18,BCLL19}. This issue is fixed by modifying both the Bregman 
divergence and the control function of the mirror descent dynamic. The latter damps down the 
movement of the algorithm when it incurs service cost at a rate close to $0$.

We obtain the following result for the evolving tree game.
\begin{theorem}\label{thm:evolvingIntro}
There exists a $O(k^2\log d_{\max})$-competitive randomized algorithm for the depth $k$ evolving tree game, where $d_{\max}$ is the maximum degree of a node at any point in the game.
\end{theorem}

In the competitive ratio, one factor $k$ comes from the maximal depth $k$ of the tree. 
The other factor $k\log d_{\max}$ is due to the fact that the perturbation of the entropy can be as small as $O(1/d_{\max}^{k})$, on account of having to set it consistently against
unknown future topology changes. (One way to think about it is that we are handling in the
background a complete tree of depth $k$, only part of which is currently revealed. In fact,
our bounds do not require that the number of leaves be restricted to $k$, as is the case with
layered graph traversal.)

Since the reduction in Theorem~\ref{thm:reductionShort} only requires a \emph{binary} tree, the factor $\log d_{\max}$ becomes a constant for the layered graph traversal problem.
\begin{theorem}\label{thm:mainIntro}
There is a randomized $O(k^2)$-competitive algorithm for traversing width $k$ layered graphs, as well as the equivalent problem of chasing sets of cardinality $k$ in any metric space.
\end{theorem}

We note that implementing the mirror descent 
approach in evolving trees is a major challenge to the design and analysis of online algorithms 
for online problems in metric spaces (e.g., the $k$-server problem, see~\cite{Lee18}, where an 
approach based on mirror descent in an evolving tree is also studied). Our ideas 
may prove applicable to other problems.

\subsection{Organization}
The rest of this paper is organized as follows. In Section~\ref{sec: evolving tree} we define and
analyze the evolving tree game. In Section~\ref{sec: motivation} we motivate the evolving tree
algorithm and analysis. In Section~\ref{sec: applications} we discuss the application to
layered graph traversal/small set chasing.

\section{The Evolving Tree Game}\label{sec: evolving tree}

For a rooted edge-weighted tree $T=(V,E)$, we denote by $r$ its root, by $V^0:= V\setminus\{r\}$ the set of 
non-root vertices and by $\cL\subseteq V^0$ the set of leaves. For $u\in V^0$, we denote by $p_u$ the parent 
of $u$ and by $w_u$ the length of the edge connecting $u$ to $p_u$.

The evolving tree game is a two person continuous time game between an adversary and an algorithm. 
The adversary grows a rooted edge-weighted tree $T=(V,E)$ of bounded degree. Without loss of generality, 
we enforce that the root $r$ always has degree $1$, and we denote its single child by $c_r$. Initially 
$V = \{r,c_r\}$, and the two nodes are connected by a zero-weight edge. The root $r$ will be fixed throughout 
the game, but the identity of its child $c_r$ may change as the game progresses. The game has continuous 
steps and discrete steps  (see Fig.~\ref{fig:steps}):. 
\begin{itemize}
\item {\bf Continuous step:} The adversary picks a leaf $\ell$ and increases the weight $w_\ell$ of the 
         edge incident on $\ell$ at a fixed rate of $w'_\ell = 1$ for a finite time interval.
\item {\bf Discrete step:} There are two types of discrete steps:
         \begin{itemize}
         \item {\bf Delete step:} The adversary chooses a leaf $\ell\in\cL$, $\ell\ne c_r$, 
                  and deletes $\ell$ and its incident edge from $T$. If the parent $p_\ell$ of $\ell$ 
                  remains with a single child $c$, the adversary smooths $T$ at $p_\ell$ as follows: 
                  it merges the two edges $\{c,p_\ell\}$ and $\{p_\ell,p_{p_\ell}\}$ into a single edge 
                  $\{c,p_{p_\ell}\}$, removing the vertex $p_\ell$, and assigns $w_c\gets w_c + w_{p_\ell}$.
         \item {\bf Fork step:} The adversary generates two or more new
                  nodes and connects all of them to an existing leaf $\ell\in\cL$ with edges of weight $0$.
                  Notice that this removes $\ell$ from $\cL$ and adds the new nodes to $\cL$.
         \end{itemize}
\end{itemize}
The continuous and discrete steps may be interleaved arbitrarily by the adversary.

A pure strategy of the algorithm maps the timeline to a leaf of $T$ that exists at that time. Thus, the start of
the game (time $0$ of step $1$) is
mapped to $c_r$, and at all times the algorithm occupies a leaf and may move from leaf to leaf. If the algorithm 
occupies a leaf $\ell$ continuously while $w_\ell$ grows, it pays the increase in weight (we call this the \emph{service cost}). If the algorithm moves 
from a leaf $\ell_1$ to a leaf $\ell_2$, it pays the total weight of the path in $T$ between $\ell_1$ and $\ell_2$ at 
the time of the move (we call this the \emph{movement cost}). A mixed strategy/randomized algorithm is, as usual, a probability distribution over pure 
strategies. A fractional strategy maps the timeline to probability distributions over the existing leaves. Writing 
$x_u$ for the total probability of the leaves in the subtree of $u$, this means that a fractional strategy maintains 
at all times a point in the changing polytope
\begin{equation}\label{eq: polytope}
K(T):=\left\{x\in\R_+^{V^0}\colon \sum_{v \colon p_v = u} x_v = x_u\,\forall u\in V\setminus\cL\right\},
\end{equation}
where we view $x_r:=1$ as a constant.

In our reduction from layered graph traversal to the evolving tree game, which we prove formally in Section~\ref{sec: applications}, the leaves of the evolving tree correspond to vertices in the last revealed layer of the layered graph traversal instance. When a new layer is revealed, the edges from the old layer can be modelled through a sequence of continuous steps and fork steps. Delete steps are used to handle the case that a vertex has no successor in the next layer.

Notice that the tree $T$, the weight function $w$, the point $x\in K(T)$, and the derived parameters are all
functions of the adversary's step and the time $t$ within a continuous step. Thus, we shall use henceforth 
the following notation: $T(j,0)$, $w(j,0)$, $x(j,0)$, etc. to denote the values of these parameters at the start of
step $j$ of the adversary. If step $j$ is a continuous step of duration $\tau$, then for $t\in [0,\tau]$, we use
$T(j,t)$, $w(j,t)$, $x(j,t)$, etc. to denote the values of these parameters at time $t$ since the start of step $j$.
If it is not required to mention the parameters $(j,t)$ for clarity, we omit them from our notation.

We require that for a fixed continuous step $j$, the function $x(j,\cdot)\colon[0,\tau]\to K(T)$ is absolutely
continuous, and hence differentiable almost everywhere. Notice that the polytope $K(T)$ is fixed during a 
continuous step, so this requirement is well-defined. We denote by $x'$ the derivative of $x$ with respect
to $t$, and similarly we denote by $w'$ the derivative of $w$ with respect to $t$. The cost of the algorithm 
during the continuous step $j$ is
$$
\int_0^\tau \sum_{v\in V^0} \left(w'_v(j,t) x_v(j,t) + w_v(j,t) \left|x'_v(j,t)\right|\right) dt.
$$
Notice that the first summand (the \emph{service cost}) is non-zero only at the single leaf $\ell$ for 
which $w_\ell(j,t)$ is growing.

In a discrete step $j$, the topology of the tree and thus the polytope $K(T)$ changes. The old tree is $T(j,0)$
and the new tree is $T(j+1,0)$. In a delete step, when a leaf $\ell$ is deleted, the algorithm first has to move 
from its old position $x(j,0)\in K(T(j,0))$ to some position $x\in K(T(j,0))$ with $x_\ell=0$. The cost of moving 
from $x(j,0)$ to $x$ is given by
$$
\sum_{v\in V^0} w_v(j,0) \left|x_v(j,0) - x_v\right|.
$$
The new state $x(j+1,0)$ is the projection of $x$ onto the new polytope $K(T(j+1,0))$, where the $\ell$-coordinate 
and possibly (if smoothing happens) the $p_\ell$-coordinate are removed.

In a fork step, the algorithm chooses as its new position any point in $K(T(j+1,0))$ whose projection onto $K(T(j,0))$ 
is the old position of the algorithm. No cost is incurred here (since the new leaves are appended at distance $0$).

The following lemma is analogous to Lemma~\ref{lm: fractional to mixed}. Its proof is very
similar. It is omitted here; we actually do not need this claim to prove the main result of the
paper. 
\begin{lemma}\label{lm: DTG frac to mix}
For every fractional strategy of the algorithm there is a mixed strategy incurring the same cost in expectation.
\end{lemma}

Our main result in this section is the following theorem. By letting $\epsilon\to 0$, it implies Theorem~\ref{thm:evolvingIntro}.
\begin{theorem}\label{thm: main}
For every $k\in\NN$ and for every $\eps > 0$ there exists a fractional strategy of the algorithm 
with the following performance guarantee.  For every pure strategy of the adversary that grows
trees of depth at most $k$, the cost $C$ of the algorithm satisfies 
$$
C\le O(k^2\log d_{\max})\cdot(\opt + \eps),
$$
where $\opt$ is the minimum distance in the final tree from the root to a leaf, and $d_{\max}$ 
is the maximum degree of a node at any point during the game.                                                                                                          
\end{theorem}
Notice that for every strategy of the adversary, there exists a pure strategy that pays 
exactly $\opt$ service cost and zero movement cost. We will refer to this strategy as the {\em optimal play}. The algorithm 
cannot in general choose the optimal play because the evolution of the tree only gets revealed 
step-by-step.

\subsection{Additional notation}

Let $j$ be any step, and let $t$ be a time in that step (so, if $j$ is a discrete step then $t=0$, and
if $j$ is a continuous step of duration $\tau$ then $t\in [0,\tau]$).
For a vertex $u\in V(j,t)$ we denote by $h_u(j,t)$ its combinatorial depth, i.e., the number of edges on the 
path from $r$ to $u$ in $T(j,t)$. Instead of the actual edge weights $w_u(j,t)$, our algorithm will be based 
on revised edge weights defined as
\begin{align*}
\tilde w_u(j,t) := \frac{2k-1}{2k-h_u(j,t)}\left(w_u(j,t) + \eps 2^{-j_u} \right),
\end{align*}
where $j_u\in\mathbb N$ is the step number when $u$ was created (or $0$ if $u$ existed in the initial tree).
The purpose of the term 
$\eps2^{-j_u}$ is to ensure that $\tilde w_u(j,t)$ is strictly positive, and is essentially negligible in calculations by choosing $\eps$ very close to $0$. The multiplication by $\frac{2k-1}{2k-h_u(j,t)}$ means that the revised weights overestimate the true weights by a factor at most $2$, which decreases over time as the depth of vertices decreases due to merge steps. This decrease will counteract an otherwise adverse effect of merge steps on the ``weighted depth'' potential function that we will use in the analysis.

For $u\in V^0(j,t)$, we also define a shift parameter by induction on $h_u(j,t)$, as follows. For 
$u=c_r(j,t)$, $\delta_u(j,t) = 1$. For other $u$, $\delta_u(j,t) = \delta_{p_u}(j,t) / (d_{p_u}(j,t)-1)$, where 
$p_u = p_u(j,t)$, and $d_{p_u}(j,t)$ is the degree of $p_u$ in $T(j,t)$ (i.e., $d_{p_u}(j,t)-1$ is the 
number of children of $p_u$ in $T(j,t)$; note that every non-leaf node in $V^0(j,t)$ has degree at least $3$). 
Observe that by definition $\delta(j,t)\in K(T(j,t))$. As mentioned earlier, we often omit the parameters
$(j,t)$ from our notation, unless they are required for clarity.

\subsection{The algorithm}\label{sec:algo}

We consider four distinct types of steps: continuous steps, fork steps, deadend steps, and merge steps.
A delete step is implemented by executing a deadend step, and if needed followed by a merge step. It is 
convenient to examine the two operations required to implement a delete step separately.

\paragraph*{\bf Continuous step} 
In a continuous step, the weight $w_\ell$ of some leaf $\ell\in\cL$ grows continuously (and thus $\tilde w_\ell'> 0$). 
In this case, for $u\in V^0$ we update the fractional mass in the subtree below $u$ at rate
\begin{equation}\label{eq: dynamic}
x_u' = -\,\, \frac{2 x_u}{\tilde w_u}\tilde w_u' + \frac{x_u+\delta_u}{\tilde w_u}\left(\lambda_{p_u} - \lambda_{u}\right),
\end{equation}
where $\lambda_u=0$ for $u\in\cL$ and $\lambda_u\ge 0$ for $u\in V\setminus\cL$ are chosen such that $x$ 
remains in the polytope $K(T)$. 
We will show in Section~\ref{sec:MDExistence}  that such $\lambda$ exists (as a function of time). The first negative term in~\ref{eq: dynamic} causes a decrease of mass at the leaf $\ell$ whose edge is growing. To prevent this from reducing the mass to less than $1$, the other term involving $\lambda$ causes a rebalancing of mass, with the overall effect being that mass moves from $\ell$ to the other leaves. \footnote{The 
dynamic of $x$ corresponds to running mirror descent with regularizer 
$\Phi_t(z) = \sum_{u\in V} \tilde w_u(z_u+\delta_u)\log(z_u+\delta_u)$, using 
the growth rate $\tilde w'$ of the approximate weights as cost function, and scaling the rate of movement by a factor 
$\frac{2x_{\ell}}{x_{\ell}+\delta_{\ell}}$ when $\ell$ is the leaf whose edge grows. See Section~\ref{sec: motivation}.}

\paragraph*{\bf Fork step} 
In a fork step, new leaves $\ell_1,\ell_2,\dots,\ell_q$ (for some $q\ge 2$) are spawned as children of a leaf $u$ 
(so that $u$ is no longer a leaf). They are ``born'' with $w_{\ell_1}=w_{\ell_2}=\cdots = w_{\ell_q}=0$ and 
$x_{\ell_1}=x_{\ell_2}=\cdots =x_{\ell_q}=x_u/q$.

\paragraph*{\bf Deadend step} 
In a deadend step, we delete a leaf $\ell\ne c_r$. To achieve this, we first compute the limit of a continuous step where 
the weight $\tilde w_\ell$ grows to infinity, ensuring that the mass $x_\ell$ tends to $0$. This, of course, changes the
mass at other nodes, and we update $x$ to be the limit of this process. Then, we remove the leaf $\ell$ along with 
the edge $\{\ell,p_\ell\}$ from the tree. Notice that this changes the degree $d_{p_\ell}$. Therefore, it triggers a discontinuous 
change in the shift parameter $\delta_u$ for every vertex $u$ that is a descendant of $p_\ell$.

\paragraph*{\bf Merge step} 
A merge step immediately follows a deadend step if, after the removal of the edge $\{\ell,p_\ell\}$, the vertex $v=p_\ell$ 
has only one child $c$ left. Notice that $v\ne r$. We merge the two edges $\{c,v\}$ of weight $w_c$ and $\{v, p_v\}$ 
of weight $w_v$ into a single edge $\{c,p_v\}$ of weight $w_c+w_v$. The two edges that were merged and the vertex
$v$ are removed from $T$. This decrements by $1$ the combinatorial depth $h_u$ of every vertex $u$ in the subtree 
rooted at $c$. Thus, it triggers a discontinuous change in the revised weight $\tilde w_u$, for every vertex $u$ in this 
subtree.

\subsection{Competitive analysis}

The analysis of the algorithm is based on a potential function argument. Let $y\in K(T)$ denote the state of the
optimal play. Note that as the optimal play is a pure strategy, the vector $y$ is simply the indicator
function for the nodes on some root-to-leaf path. We define a potential function $P = P_{k,T,w}(x,y)$, 
where $x\in K(T)$ and we prove that the algorithm's cost plus the change in $P$ is at most $O(k^2\log d_{\max})$ times 
the optimal cost, where $d_{\max}$ is the maximum degree of a node of $T$. This, along with the fact that 
$P$ is bounded, implies $O(k^2\log d_{\max})$ competitiveness. In Section~\ref{sec: motivation} we motivate the construction 
of the potential function $P$. Here, we simply define it as follows:
\begin{equation}\label{eq: potential}
P := 2\sum_{u\in V^0}\tilde w_u \left(4k y_u \log \frac{1+\delta_u}{x_u+\delta_u} + (2k-h_u)x_u\right).
\end{equation}
We now consider the cost and potential change for each of the different steps separately.

\subsubsection{Continuous step}\label{sec:growth}

\paragraph*{\bf Bounding the algorithm's cost}
Let $\ell$ be the leaf whose weight $w_\ell$ is growing, and recall that $c_r$ is the current neighbor of the root $r$. By 
definition of the game, the algorithm pays two types of cost. Firstly, it pays for the mass $x_\ell$ at the leaf $\ell$ moving 
away from the root at rate $w_\ell'$. Secondly, it pays for moving mass from $\ell$ to other leaves. Notice that $x_{c_r} = 1$
stays fixed. Let $C = C(j,t)$ denote the total cost that the algorithm accumulates in
the current step $j$, up to the current time $t$.
\begin{lemma}\label{lm: growth cost}
The rate at which $C$ increases with $t$ is
$$
C' \le 3\tilde w_\ell'x_\ell + 2\sum_{u\in V^0} (x_u+\delta_u)\lambda_{u}.
$$
\end{lemma}

\begin{proof}
We have
\begin{eqnarray}
         C' 
& = & w_\ell'x_\ell\, + \sum_{u\in V^0\setminus \{c_r\}} w_u |x_u'| \nonumber\\
&\le& \tilde w_\ell'x_\ell\, + \sum_{u\in V^0\setminus \{c_r\}} \tilde w_u |x_u'|\label{eq: growth w to tilde w}\\
& = & \tilde w_\ell'x_\ell\, +\!\!\!\!\! \sum_{u\in V^0\setminus \{c_r\}} \!\!\!\!\! \left|-2x_u \tilde w_u' + 
                                (x_u+\delta_u)(\lambda_{p_u} - \lambda_u)\right|\label{eq: growth movement}\\
&\le& 3\tilde w_\ell'x_\ell\, + \sum_{u\in V^0\setminus \{c_r\}} (x_u+\delta_u)(\lambda_{p_u} + \lambda_u)\label{eq: growth triangle} \\
&\le& 3\tilde w_\ell'x_\ell + 2\sum_{u\in V^0} (x_u+\delta_u)\lambda_{u}.\label{eq:growthCost}
\end{eqnarray}
Inequality~\eqref{eq: growth w to tilde w} uses the fact that $w_u\le \tilde w_u$ and $w_\ell'\le \tilde w_\ell'$. 
Equation~\eqref{eq: growth movement} uses the definition of the dynamic in Equation~\eqref{eq: dynamic}.
Inequality~\eqref{eq: growth triangle} uses the triangle inequality.
Finally, Inequality~\eqref{eq:growthCost} uses the fact that $x,\delta\in K(T)$, so 
$\sum_{v \colon p_v = u} (x_v+\delta_v) = x_u+\delta_u$ for all $u\in V\setminus\cL$.
\end{proof}

\paragraph*{\bf Change of potential}
We decompose the potential function as $P= 4kD - 2\Psi$, where
\begin{align*}
D := \sum_{u\in V^0}\tilde w_u\left(2y_u\log\frac{1+\delta_u}{x_u+\delta_u} \,\,\,+\,\,\, x_u\right)
\end{align*}
is a variant of the Bregman divergence (i.e., multiscale KL-divergence) between $y$ and $x$, and
\begin{align*}
\Psi:= \sum_{u\in V^0} h_u \tilde w_u x_u
\end{align*}
corresponds to the weighted depth potential from~\cite{BCLL19,BCLLM18}.

We first analyze the rate of change of $\Psi$, which allows to charge the algorithm's total cost to only its service cost.
\begin{lemma}\label{lm: Psi change}
The rate at which $\Psi$ changes satisfies:
$$
\Psi' \ge -k \tilde w_\ell' x_\ell + \sum_{u\in V^0} \lambda_u (x_u+\delta_u).
$$
\end{lemma}

\begin{proof}
We have
\begin{eqnarray}
\Psi' 
       & = & h_\ell \tilde w_\ell' x_\ell + \sum_{u\in V^0\setminus\{c_r\}} h_u \tilde w_u x_u'\nonumber\\
       & = & h_\ell \tilde w_\ell' x_\ell+\!\!\sum_{u\in V^0\setminus\{c_r\}} \!\!h_u \left(-2x_u \tilde w_u' + (x_u+\delta_u)(\lambda_{p_u}-\lambda_{u})\right)\nonumber\\
       & = & - h_\ell \tilde w_\ell' x_\ell + \sum_{u\in V^0\setminus\{c_r\}} h_u (x_u+\delta_u)(\lambda_{p_u}-\lambda_{u})\nonumber\\
       &\ge& - k \tilde w_\ell' x_\ell + \!\!\!
               \sum_{u\in V^0} \!\!\lambda_u \left((h_u\!+\!1)\!\!\!\!\!\!\sum_{v\colon p_v=u} \!\!\!\!\! (x_v\!+\!\delta_v) \!-\! h_u (x_u\!+\!\delta_u)\right)\label{eq: psi ineq}\\
       & = & - k \tilde w_\ell' x_\ell + \sum_{u\in V^0} \lambda_u (x_u+\delta_u).\label{eq:growthWdepth}
\end{eqnarray}
Here, Inequality~\eqref{eq: psi ineq} uses the fact that $h_\ell\le k$ and, for $u=p_v$, $h_v=h_u+1$.
Equation~\eqref{eq:growthWdepth} uses the previously noted fact that, as $x,\delta\in K(T)$, then for all $u\notin \cL$,
$\sum_{v \colon p_v = u} (x_v+\delta_v) = x_u+\delta_u$ (and if $u\in\cL$, then $\lambda_u=0$).
\end{proof}

Next, we analyze the rate of change of $D$, which allows us to charge algorithm's service cost against the service cost of the optimal play.
\begin{lemma}\label{lm: D change}
The rate at which $D$ changes satisfies:
$$
D' \le -\tilde w_\ell'x_\ell + 2(2+k\log d_{\max}) y_\ell\tilde w_\ell'.
$$
\end{lemma}

\begin{proof}
We have
\begin{eqnarray}
          D' 
& = & \tilde w_\ell' \left(2y_\ell\log\frac{1+\delta_\ell}{x_\ell+\delta_\ell}+x_\ell\right) + 
 \!\!\!\!\!\!\sum_{u\in V^0\setminus\{c_r\}} \!\!\!\!\!\!\tilde w_ux_u'\left(\frac{-2y_u}{x_u+\delta_u} + 1\right)\nonumber\\
& = & \tilde w_\ell' \left(2y_\ell\log\frac{1+\delta_\ell}{x_\ell+\delta_\ell}+x_\ell\right) + \nonumber \\
&    & + \!\!\!\!\!\!\!\! \sum_{u\in V^0\setminus\{c_r\}} \!\!\!\!\!\!\!\!\left(-2x_u \tilde w_u' + 
                              (x_u+\delta_u)(\lambda_{p_u} - \lambda_u)\right)\left(\frac{-2y_u}{x_u+\delta_u} + 1\right)\nonumber\\
& = & -\tilde w_\ell'x_\ell + 2y_\ell\tilde w_\ell' \left(\log\frac{1+\delta_\ell}{x_\ell+ \delta_\ell} + 
                    \frac{2 x_\ell}{x_\ell+\delta_\ell}\right) +  \nonumber \\
&    & + \sum_{u\in V^0\setminus \{c_r\}}(\lambda_{p_u}-\lambda_u)(-2y_u+x_u+\delta_u)\nonumber\\
&\le& -\tilde w_\ell'x_\ell + 2 y_\ell\tilde w_\ell'(2+k\log d_{\max}) + \nonumber \\
&    & +\sum_{u\in V^0} \lambda_u \left(2y_u-x_u-\delta_u - \!\!\!\!\sum_{v\colon p_v=u}\!\!\!\!(2y_v-x_v-\delta_v)\right)\label{eq: D ineq}\\
& = & -\tilde w_\ell'x_\ell + 2 y_\ell\tilde w_\ell'(2+k\log d_{\max}).\label{eq:growthBregman}
\end{eqnarray}
Inequality~\eqref{eq: D ineq} uses the fact that $\delta_\ell\ge (d_{\max})^{1-h_\ell}\ge (d_{\max})^{1-k}$ and 
$2y_{c_r}-x_{c_r}-\delta_{c_r}=0$. Equation~\eqref{eq:growthBregman} uses the fact that $x,y,\delta\in K(T)$, hence 
for $u\in V\setminus\cL$, 
$$
2y_u-x_u-\delta_u = \sum_{v\colon p_v=u} (2y_v-x_v-\delta_v)
$$ 
(and for $u\in\cL$, $\lambda_u=0$).
\end{proof}

We obtain the following lemma, which bounds the algorithm's cost and change in potential against the service cost of the optimal play.
\begin{lemma}\label{lm: P change}
For every $k\ge 2$, it holds that 
$$
C' + P' \le O(k^2\log d_{\max}) w_\ell' y_\ell.
$$
\end{lemma}

\begin{proof}
Combine Equations~\eqref{eq:growthCost},~\eqref{eq:growthWdepth}, and~\eqref{eq:growthBregman}, 
and recall that $P=4kD-2\Psi$. We get
\begin{eqnarray}
C' +P' & \le & (2k+3-4k)\tilde w_\ell' x_\ell  + 8k(2+k\log d_{\max})\tilde w_\ell' y_\ell \nonumber \\
          & \le  & O(k^2\log d_{\max}) w_\ell' y_\ell,\label{eq:growthMainBound}
\end{eqnarray}
where in the last inequality we use $\tilde w_\ell' < 2w_\ell'$.
\end{proof}

\subsubsection{Fork step}
\quad
\\

Fork steps may increase the value of the potential function $P$, because the new edges
have revised weight $> 0$. The following lemma bounds this increase.
\begin{lemma}\label{lm: fork cost}
The combined contribution of all fork steps to $P$ is an increase of at most $\eps \cdot O(k^2\log d_{\max})$.
\end{lemma}

\begin{proof}
Consider a fork step that attaches new leaves $\ell_1,\dots,\ell_q$ to a leaf $u$.
The new leaves are born with revised edge weights 
$\frac{2k-1}{2k-h_u-1}\eps 2^{-j}\le \eps 2^{-j+1}$, where $j$ is the current step number. Since $\sum_{i=1}^q y_{\ell_i}=y_u\le 1$ and $\sum_{i=1}^q x_{\ell_i}=x_u\le 1$, the change $\Delta P$ in $P$ satisfies
\begin{eqnarray*}
          \Delta P 
&\le&  \eps 2^{-j+2}\cdot\left(4k\log\frac{1+\delta_u/q}{\delta_u/q} + 2k-h_{u}-1\right) \\
&\le& \eps 2^{-j+2}\cdot (2k+4k^2\log d_{\max}),
\end{eqnarray*}
where the last inequality follows from $\delta_u/q\ge (d_{\max})^{1-k}$.
As the step number $j$ is different in all fork steps, the total cost of all fork steps is
at most 
$$
\eps\cdot(2k+4k^2\log d_{\max})\sum_{j=1}^{\infty} 2^{-j+2} = \eps\cdot O(k^2\log d_{\max}),
$$
as claimed.
\end{proof}

\subsubsection{Deadend step}
\quad
\\

Recall that when a leaf $\ell$ is deleted, we first compute the limit of a continuous step as the weight 
$\tilde w_\ell$ grows to infinity. Let $\bar x$ be the mass distribution that the algorithm converges to 
when $\tilde w_\ell$ approaches infinity.
\begin{lemma}\label{lm: limit 0}
The limit $\bar x$ satisfies $\bar x_{\ell}=0$. Hence, $\bar x$ with the $\ell$-coordinate removed is a 
valid mass distribution in the new polytope $K(T)$. Also, a deadend step decreases $P$ by at least
the cost the algorithm incurs to move to $\bar x$. 
\end{lemma}

\begin{proof}
Note that $y_\ell=0$ for the ``dying'' leaf $\ell$. Thus, by Lemma~\ref{lm: P change}, the cost of the algorithm during
the growth of $\tilde w_\ell$ is bounded by the decrease of $P$ during that time. Clearly, $P$ can only decrease by 
a finite amount (as it remains non-negative) and thus the algorithm's cost is finitely bounded. But this means that the 
mass at $\ell$ must tend to $0$, since otherwise the service cost would be infinite. Moreover, notice that the growth
of $\tilde w_\ell$ is just a simulation and the algorithm doesn't pay the service cost, only the cost of moving from its
state $x$ at the start of the simulation to the limit state $\bar x$. However, this movement cost is at most the total
cost to the algorithm during the simulation, and $P$ decreases by at least the total cost. Finally, at $\bar x$, the term 
in $P$ for $\ell$ equals $0$, so removing it does not increase $P$. Also, for every vertex $u$ in a subtree rooted at a 
sibling of $\ell$ the term $\delta_u$ increases (as the degree $d_{p_\ell}$ decreases by $1$). However, this too cannot increase 
$P$ (as $x_u\le 1$).
\end{proof}

\subsubsection{Merge step}

\begin{lemma}\label{lm: merge}
A merge step does not increase $P$.
\end{lemma}

\begin{proof}
Let $j$ be the step number in which the merge happens.
Substituting the expression for the revised weights, the potential $P$ can be written as
$$
P = 2\!\!\!\sum_{u\in V^0}\!\!\!\left(w_u + \eps 2^{-j_u} \right) \left(\frac{2k-1}{2k-h_u}4k y_u \log \frac{1+\delta_u}{x_u+\delta_u} + (2k-1)x_u\right).
$$
Consider the two edges $\{c,v\}$ and $\{v, p_v\}$ that are to be merged, where $v=p_c(j,0)$. Firstly, 
for each vertex $u$ in the subtree of $c$ (including $c$ itself), its depth $h_u$ decreases by $1$.
This cannot increase $P$. Notice also that as $d_v = 2$, we have $\delta_c = \delta_v$ and the merge
does not change any $\delta_u$. The new value $h_c(j+1,0)$ equals the old value $h_v(j,0)$. Note 
also that  $y_c=y_v$ and $x_c=x_v$ because $c$ is the only child of $v$. Thus, merging the two 
edges of lengths $w_c$ and $w_v$ into a single edge of length $w_c+w_v$, and removing vertex 
$v$, only leads to a further decrease in $P$ resulting from the disappearance of the $2^{-j_v}$ term.
\end{proof}

\subsubsection{Putting it together}
\quad

\begin{proofof}{Theorem~\ref{thm: main}}
By Lemmas~\ref{lm: P change},~\ref{lm: fork cost},~\ref{lm: limit 0} and~\ref{lm: merge},
$$
C\le O(k^2\log d_{\max}) \opt + P_0 - P_f + \eps\cdot O(k^2\log d_{\max}),
$$
where $P_0$ and $P_f$ are the initial and final value of $P$, respectively. Now, observe that 
$P_0 = \eps\cdot O(k)$ and $P_f\ge 0$.
\end{proofof}

\section{Derivation of Algorithm and Potential Function}\label{sec: motivation}

We now describe how we derived the algorithm and potential function from the last section, and justify the existence of $\lambda$.

\subsection{Online mirror descent}
Our algorithm is based on the online mirror descent framework of \cite{BCLLM18,BCLL19}. In general, an algorithm in this framework is specified by a convex body $K\subset \R^n$, a suitable strongly convex function $\Phi\colon K\to \R$ (called \emph{regularizer}) and a map $f\colon [0,\infty)\times K\to \R^n$ (called \emph{control function}). The algorithm corresponding to $K$, $\Phi$ and $f$ is the (usually unique) solution $x\colon[0,\infty)\to K$ to the following differential inclusion:
\begin{align}
	\nabla^2\Phi(x(t))\cdot x'(t) \in f(t,x(t)) - N_K(x(t)),\label{eq:MD}
\end{align}
where $\nabla^2\Phi(x)$ denotes the Hessian of $\Phi$ at $x$,
\[N_K(x):=\{\mu\in\R^n\colon \langle \mu, y-x\rangle \le 0, \,\forall y\in K\}\]
is the normal cone of $K$ at $x(t)$, and the right-hand side in~\ref{eq:MD} indicates the
set $\left\{f(t,x(t))-\mu\colon \mu\in N_K(x(t))\right\}$.
Intuitively, \eqref{eq:MD} means that $x$ tries to move in direction $f(t,x(t))$, with the normal cone term $N_K(x(t))$ ensuring that $x(t)\in K$ can be maintained, and multiplication by the positive definite matrix $\nabla^2\Phi(x(t))$ corresponding to a distortion of the direction in which $x$ is moving. This corresponds to the continuous-time version of mirror descent. The left hand side of \eqref{eq:MD} is nothing but $\frac{d\nabla\Phi(x(t))}{dt}$.

A benefit of the online mirror descent framework is that there exists a default potential function for its analysis, namely the Bregman divergence associated to $\Phi$, defined as
\begin{align*}
	D_\Phi(y\| x):=\Phi(y)-\Phi(x)+\langle \nabla \Phi(x), x-y\rangle
\end{align*}
for $x,y\in K$. Plugging in $x=x(t)$, the change of the Bregman divergence as a function of time is
\begin{align}
	\frac{d}{dt}D_\Phi(y\| x(t)) 
	&= \langle \nabla^2\Phi(x(t))\cdot x'(t), x(t)-y\rangle\label{eq:chainRule}\\
&= \langle f(t,x(t))-\mu(t), x(t)-y\rangle\nonumber \\
& \qquad\qquad\text{for some $\mu(t)\in N_K(x(t))$}\label{eq:plugMD}\\
&\le \langle f(t,x(t)), x(t)-y\rangle,\label{eq:BregmanBound}
\end{align}
where equation \eqref{eq:chainRule} follows from the definition of $D_\Phi$ and the chain rule, \eqref{eq:plugMD} follows from~\eqref{eq:MD}, 
and~\eqref{eq:BregmanBound} follows from the definition of $N_K(x(t))$.

\subsection{Charging service cost for evolving trees}

In the evolving tree game, we have $K=K(T)$. For a continuous step, it would seem natural to choose $f(t,x(t))=-w'(t)$, so that \eqref{eq:BregmanBound} implies that the online service cost $\langle w'(t), x(t)\rangle$ plus change in the potential $D_\Phi(y\| x(t))$ is at most the offline service cost $\langle w'(t), y\rangle$. For the regularizer $\Phi$ (which should be chosen in a way that also allows to bound the movement cost later), the choice analogous to \cite{BCLL19} would be
\begin{align*}
\Phi(x):= \sum_{u\in V^0} w_u(x_u+\delta_u)\log(x_u+\delta_u).
\end{align*}

However, since $\Phi$ (and thus $D_\Phi$) depends on $w$, the evolution of $w$ leads to an additional change of $D_\Phi$, which the bound~\eqref{eq:BregmanBound} does not account for as it assumes the regularizer $\Phi$ to be fixed. To determine this additional change, first observe that by a simple calculation
\begin{align*}
	D_\Phi(y\|x) = \sum_{u\in V^0} w_u\left[(y_u+\delta_u)\log \frac{y_u+\delta_u}{x_u+\delta_u} + x_u-y_u\right].
\end{align*}
When $w_\ell$ increases at rate $1$, this potential increases at rate 
$$
(y_\ell+\delta_\ell)\log \frac{y_\ell+\delta_\ell}{x_\ell+\delta_\ell} + x_\ell-y_\ell.
$$
The good news is that the part 
$$
y_\ell\log \frac{y_\ell+\delta_\ell}{x_\ell+\delta_\ell}\le y_\ell \cdot O\left(\log\frac{1}{\delta_\ell}\right)\le O(k) y_\ell
$$ 
can be charged to the offline service cost, which increases at rate $y_\ell$. The term $-y_\ell$ also does no harm as it is non-positive. The term $x_\ell$ might seem to be a slight worry, because it is equal to the online service cost, which is precisely the quantity that the change in potential is supposed to cancel. It means that effectively we have to cancel two times the online service cost, which can be achieved by accelerating the movement of the algorithm by a factor $2$ (by multiplying the control function by a factor $2$). The main worry is the remaining term $\delta_\ell\log \frac{y_\ell+\delta_\ell}{x_\ell+\delta_\ell}$, which does not seem controllable. We would therefore prefer to have a potential that has the $\delta_u$ terms only inside but not outside the $\log$.

Removing this term (and, for simplicity, omitting the $y_u$ at the end of the potential, which does not play any important role), our desired potential would then be a sum of the following two terms $L(t)$ and $M(t)$:
\begin{align*}
	L(t) &:= \sum_{u\in V^0} w_u(t) y_u\log \frac{y_u+\delta_u}{x_u(t)+\delta_u}\\
	M(t) &:= \sum_{u\in V^0} w_u(t) x_u(t)
\end{align*}

Let us study again why these terms are useful as part of the classical Bregman divergence potential by calculating their change. Dropping $t$ from the notation, and using that $\nabla^2\Phi(x)$ is the diagonal matrix with entries $\frac{w_u}{x_u+\delta_u}$, we have
\begin{align*}
	L' &= \langle w',y\rangle O(k) - \langle y, \nabla^2\Phi(x)\cdot x'\rangle \\
	&= \langle w',y\rangle O(k) - \langle y, f-\mu\rangle 
\end{align*}
and
\begin{align*}
	M' &= \langle w', x\rangle + \langle w, x'\rangle\\
	&= \langle w', x\rangle + \langle x+\delta, \nabla^2\Phi(x) \cdot x'\rangle\\
	&= \langle w', x\rangle + \langle x+\delta,  f-\mu\rangle
\end{align*}
for some $\mu\in N_K(x)$.

For a convex body $K$ of the form $K=\{x\in\R^n\colon Ax\le b\}$ where $A\in\R^{m\times n}$, $b\in \R^n$, the normal cone is given by
\begin{align*}
N_K(x) = \{A^T\lambda\mid \lambda\in\R_+^m, \langle \lambda, Ax-b\rangle=0\}.\label{eq:normalConeGeneral}
\end{align*}
The entries of $\lambda$ are called \emph{Lagrange multipliers}. In our case, we will have $x_u>0$ for all $u\in V^0$, so the Lagrange multipliers corresponding to the constraints $x_u\ge 0$ will be zero. So the only tight constraints are the equality constraints, and the normal cone corresponding to \emph{any} such $x$ is given by
\begin{align}
	N_K(x) = \{(\lambda_{u}-\lambda_{p_u})_{u\in V^0}\mid \lambda\in\R^{V}, \lambda_u=0\,\forall u\in\cL\}.\label{eq:normalCone}
\end{align}
Since $\delta\in K$ and $\delta_u>0$ for all $u$, we thus have $N_K(x)=N_K(\delta)$. Hence, we can cancel the $\mu$ terms in $L'$ and $M'$ by taking the potential $D=2L+M$, so that
\begin{align*}
	D'=2L' + M' &= \langle w',y\rangle O(k) + \langle w', x\rangle + \langle x+\delta - 2y,  f-\mu\rangle\\
	&\le \langle w',y\rangle O(k) + \langle w', x\rangle + \langle x+\delta - 2y,  f\rangle,
\end{align*}
where the inequality uses that $\mu\in N_K(x)$ and $\mu\in N_K(\delta)$. Recalling that $w_u'=\1_{u=\ell}$, and choosing $f=-2w'$ as the control function, we get
\begin{align} \label{eq:justforbelow}
	D' &\le y_\ell O(k) - x_\ell,
\end{align}
i.e., the potential charges the online service cost to $O(k)$ times the offline service cost. Indeed, the potential function $D$ used in Section~\ref{sec:growth} is given by $D=2L+M$, up to the replacement of $w$ by $\tilde w$. Moreover we note that \eqref{eq:justforbelow} remains true with the control function $f=-\frac{2x_\ell}{x_\ell+\delta_\ell}w'$, which will be helpful for the movement as we discuss next.

\subsection{Bounding movement via damped control and revised weights}

Besides the service cost, we also need to bound the movement cost. By \eqref{eq:MD} and \eqref{eq:normalCone} and since $\nabla^2\Phi(x)$ is the diagonal matrix with entries $\frac{w_u}{x_u+\delta_u}$, the movement of the algorithm satisfies
\begin{align}
	w_u x_u'&=(x_u+\delta_u)(f_u + \lambda_{p_u} - \lambda_{u})\nonumber\\
	&= -2x_u w_u' + (x_u+\delta_u)(\lambda_{p_u} - \lambda_{u}),\label{eq:boundMovementMotivation}
\end{align}
where the last equation uses $f=-\frac{2x_\ell}{x_\ell+\delta_\ell}w'$. Up to the discrepancy between $w$ and $\tilde w$, this is precisely Equation~\eqref{eq: dynamic}. Here, damping the control function $f$ by the factor $\frac{x_\ell}{x_\ell+\delta_\ell}$ is crucial: Otherwise there would be additional movement of the form $\delta_\ell w_\ell'$. Although a similar $\delta$-induced term in the movement exists also in~\cite{BCLL19}, the argument in \cite{BCLL19} to control such a term relies on $w$ being fixed and would therefore fail in our case. Scaling by $\frac{x_\ell}{x_\ell+\delta_\ell}$ prevents such movement from occurring in the first place.

To bound the movement cost,~\cite{BCLL19} employs a \emph{weighted depth potential} defined as
\begin{align*}
	\Psi = \sum_{u\in V^0}h_u w_u x_u.
\end{align*}
Our calculation in Lemma~\ref{lm: Psi change} suggests that we can use the same $\Psi$ here, choosing the overall potential function as $P= 4kD - 2\Psi$. But now the problem is that the combinatorial depths $h_u$ can change during merge steps, which would lead to an increase of the overall potential $P$. To counteract this, we use the revised weights $\tilde w_u$: The scaling by $\frac{2k-1}{2k-h_u}$ in their definition means that $\tilde w_u$ slightly increases whenever $h_u$ decreases, and overall this ensures that the potential $P$ does not increase in a merge step. Since $\frac{2k-1}{2k-h_u}\in[1,2]$, such scaling loses only a constant factor in the competitive ratio. The additional term $2^{-j_u}$ in the definition of the revised weights only serves to ensure that $\tilde w_u>0$, so that $\Phi$ is strongly convex as required by the mirror descent framework.

\subsection{Existence of the mirror descent path for time-varying $\Phi_t$}\label{sec:MDExistence}

The statement of the following lemma was proved in the full version of~\cite{BCLLM18}:
\begin{lemma}[Bubeck et al.~{\cite[Theorem~5.7 \& Lemma 5.9]{BCLLM18}}]\label{lem:existStatic}
	Let $C\subset\R^n$ be a compact, convex set, $H\colon C\to\{A\in\R^{n\times n}\mid A\succ 0\}$ continuous, $g\colon[0,\tau]\times C\to\R^n$ continuous and $y_0\in C$. Then there is an absolutely continuous solution $y\colon[0,\tau]\to C$ satisfying
	\begin{align*}
		y'(t)&\in H(y(t))\cdot(g(t,y(t))-N_C(y(t)))\\
		y(0)&=y_0.
	\end{align*}
	Moreover, the solution is unique provided $H$ is Lipschitz and $g$ locally Lipschitz.
\end{lemma}

The existence of our algorithm is justified by the following theorem, which generalizes~\cite[Theorem 2.2]{BCLLM18} to the setting where a fixed regularizer $\Phi$ is replaced by a time-varying regularizer $\Phi_t$.

\begin{theorem}\label{thm: existence of x}
	Let $K\subset \R^n$ be compact and convex, $f\colon[0,\infty)\times K\to \R^n$ continuous, $\Phi_t\colon K\to\R$ strongly convex for $t\ge 0$ and such that $(x,t)\mapsto \nabla^2\Phi_t(x)^{-1}$ is continuous. Then for any $x_0\in K$ there is an absolutely continuous solution $x\colon[0,\infty)\to K$ satisfying
	\begin{align*}
		\nabla^2\Phi_t(x(t))\cdot x'(t)&\in f(t,x(t))-N_K(x(t))\\
		x(0)&=x_0.
	\end{align*}
	If $(x,t)\mapsto\nabla^2\Phi_t(x)^{-1}$ is Lipshitz and $f$ locally Lipschitz, then the solution is unique.
\end{theorem}

\begin{proof}
	It suffices to consider a finite time interval $[0,\tau]$. We invoke Lemma~\ref{lem:existStatic} with $C=[-1,\tau+1]\times K$,
	\begin{align*}
		H(t,x)= \begin{bmatrix}
			1& 0 \\
			0 & \nabla^2\Phi_t(x))^{-1}
		\end{bmatrix},
	\end{align*}
	$g(t,(s,x))=(1,f(t,x))$ and $y_0=(0,x_0)$. Decomposing the solution as $y(t)=(s(t),x(t))$ for $s(t)\in [-1,\tau+1]$ and $x(t)\in K$, and noting that for $s(t)\in[0,\tau]$ we have $N_C(y(t))= \{0\}\times N_K(x(t))$, we get
	\begin{align*}
		s(t)&=t\\
		x'(t)&\in \nabla^2\Phi_t(x(t)))^{-1}\cdot(f(t,x(t))-N_K(x(t)))\\
		x(0)&=x_0.
	\end{align*}
	The result follows from Lemma~\ref{lem:existStatic}.
\end{proof}
For every continuous step, 
$$
\Phi_t=\sum_{u\in V^0}\tilde w_u(t)(x_u+\delta_u)\log (x_u+\delta_u)
$$ 
and 
$$
f(t,x)=-\frac{2x_\ell(t)}{x_\ell(t)+\delta_\ell} w'(t)
$$ 
satisfy the assumptions of the theorem. By the calculation in~\eqref{eq:boundMovementMotivation} (with $w$ replaced by $\tilde w$), the corresponding well-defined algorithm is the one from equation \eqref{eq: dynamic}. Note that Lagrange multipliers for the constraints $x_u\ge 0$ are indeed not needed (see below).

\paragraph*{\bf Sign of Lagrange multipliers} 
We stipulated in Section~\ref{sec:algo} that $\lambda_u\ge 0$ for $u\in V\setminus\cL$, and we do not have any Lagrange multipliers for the constraints $x_u\ge 0$. To see this, it suffices to show that $\lambda_u\ge 0$ for $u\in V\setminus\cL$ in the case that the constraints $x_u\ge 0$ are removed from $K$: If this is true, then \eqref{eq: dynamic} shows for any leaf $u\in\cL$ that $x_u'<0$ is possible only if $x_u>0$ (since $\lambda_u=0$ when $u$ is a leaf, recalling~\eqref{eq:normalCone}). Hence, $x_u\ge 0$ holds automatically for any leaf $u$, and thus also for internal vertices $u$ due to the constraints of the polytope. Consequently, we do not need Lagrange multipliers for constraints $x_u\ge 0$.

To see that $\lambda_u\ge 0$ for $u\in V\setminus\cL$, note that if we replace in $K$ the constraints $\sum_{v \colon p_v = u} x_v = x_u$ by $\sum_{v \colon p_v = u} x_v \ge x_u$, then this directly gives $\lambda_u\ge 0$ in~\eqref{eq:normalCone}. We will show (analogously to~\cite[Lemma~3.2]{BCLL19}) that these inequality constraints will actually still be satisfied with equality: Indeed, suppose $\sum_{v \colon p_v = u} x_v(t) > x_u(t)$ for some $u\in V\setminus\cL$ at some time $t$; then $\lambda_u(t)=0$ by the complementary slackness condition in~\eqref{eq:normalConeGeneral}. But then the dynamics \eqref{eq: dynamic} yields $x_u'(t)\ge 0$ (since $\tilde w_u'(t)=0$ as $u\notin\cL$) and $x_v'(t)\le 0$ for all $v$ with $p_v=u$. Thus, $\sum_{v \colon p_v = u} x_v(t) > x_u(t)$ is impossible.

\section{Reductions and Applications}\label{sec: applications}

In this section we show that layered graph traversal and small set chasing (a.k.a.
metrical service systems) reduce to the evolving tree game, thereby obtaining Theorems~\ref{thm:reductionShort} and~\ref{thm:mainIntro}.

\subsection{Layered graph traversal}

Recall the definition of the problem in the introduction. We will introduce useful
notation. Let 
$$
V_0 = \{a\}, V_1, V_2, \dots, V_n = \{b\}
$$ 
denote the layers of the
input graph $G$, in consecutive order. Let $E_1,E_2,\dots,E_n$ be the partition
of the edge set of $G$, where for every $i=1,2,\dots,n$, every edge $e\in E_i$ 
has one endpoint in $V_{i-1}$ and one endpoint in $V_i$. Also recall that 
$w: E\rightarrow\NN$ is the weight function on the edges, and 
$k = \max\{|V_i|\colon i=0,1,2,\dots,n\}$ is the {\em width} of $G$. The
input $G$ is revealed gradually to the searcher. Let 
$$
G_i = (V_0\cup V_1\cup\cdots\cup V_i, E_1\cup E_2\cup\cdots\cup E_i)
$$ 
denote the 
subgraph that is revealed up to and including step $i$. The searcher, currently at a vertex 
$v_{i-1}\in V_{i-1}$ chooses a path in $G_i$ from $v_{i-1}$ to a vertex $v_i\in V_i$. Let 
$w_{G_i}(v_{i-1},v_i)$ denote the total weight of a shortest path from $v_{i-1}$ to $v_i$ 
in $G_i$. (Clearly, the searcher has no good reason to choose a longer path.) Formally, 
a pure strategy (a.k.a. deterministic algorithm) of the searcher is a function that maps, 
for all $i=1,2,\dots$, a layered graph $G_i$ (given including its partition into a sequence
of layers) to a vertex in $V_i$ (i.e., the searcher's next move). A mixed strategy (a.k.a. 
randomized algorithm) of the searcher is a probability distribution over such functions.

\quad
\subsubsection{Fractional strategies}
\quad
\\

Given a mixed strategy $S$ of the searcher, we can define a sequence $P_0,P_1,P_2,\dots$, 
where $P_i$ is a probability distribution over $V_i$. For every $v\in V_i$, $P_i(v)$ indicates the 
probability that the searcher's mixed strategy $S$ chooses to move to $v$ (i.e., $v_i = v$). A 
fractional strategy of the searcher is a function that maps, for all $i=1,2,\dots$, a layered graph 
$G_i$ to a probability distribution $P_i$ over $V_i$. For a fractional strategy choosing probability 
distributions $P_0,P_1,P_2,\dots,P_n$, we define its cost as follows. For $i=1,2,\dots,n$, 
let $\tau_i$ be a probability distribution over $V_{i-1}\times V_i$, with marginals $P_{i-1}$ on 
$V_{i-1}$ and $P_i$ on $V_i$, that minimizes 
$$
w_{G_i,\tau_i}(P_{i-1},P_i) = \sum_{u\in V_{i-1}}\sum_{v\in V_i} w_{G_i}(u,v) \tau_i(u,v).
$$
The cost of the strategy is then defined as $\sum_{i=1}^n w_{G_i,\tau_i}(P_{i-1},P_i)$.

The following lemma can be deduced through the reduction to small set chasing discussed
later, the fact that small set chasing is a special case of metrical task systems, and a similar
known result for metrical task systems. Here we give a straightforward direct proof.
\begin{lemma}\label{lm: fractional to mixed}
For every fractional strategy of the searcher there is a mixed strategy incurring the same cost.
\end{lemma}

\begin{proof}
Fix any fractional strategy of the searcher, and suppose that the searcher plays $P_0,P_1,P_2,\dots,P_n$ 
against a strategy $G_n$ of the designer. I.e, the designer chooses the number of rounds $n$, and plays 
in round $i$ the last layer of 
$$
G_i = (\{a\}\cup V_1\cup V_2\cup\cdots\cup V_i,E_1\cup E_2\cup\cdots\cup E_i).
$$
The searcher responds with $P_i$, which is a function of $G_i$. Notice that when the designer reveals 
$G_i$, the searcher can compute $\tau_i$, because that requires only the distance functions $w_{G_i}$ 
and the marginal probability distributions $P_{i-1}$ and $P_i$. We construct a mixed strategy of the searcher 
inductively as follows. It is sufficient to define, for every round $i$, the conditional probability distribution on 
the searcher's next move $v_i\in V_i$, given any possible play so far. Initially, at the start of round $1$, the 
searcher is deterministically at $a$. Suppose that the searcher reached a vertex $v_{i-1}\in V_{i-1}$. Then, 
we set 
$$
\Pr[v_i = v\in V_i\mid v_{i-1}] = \frac{\tau_i(v_{i-1},v)}{P_{i-1}(v_{i-1})}.
$$ 
Notice that the searcher can move from $v_{i-1}$ to $v_i$ along a path in $G_i$ of length $w_{G_i}(v_{i-1},v_i)$. 

We now analyze the cost of the mixed strategy thus defined. We prove by induction over the
number of rounds that in round $i$, for every pair of vertices $u\in V_{i-1}$ and $v\in V_i$,
the probability that the searcher's chosen pure strategy (which is a random variable) reaches $v$
is $P_i(v)$ and the probability that this strategy moves from $u$ to $v$ is $\tau_i(u,v)$
(the latter assertion is required to hold for $i > 0$). The base case is $i=0$, which is trivial, 
as the searcher's initial position is $a$, $P_0(a) = 1$, and the statement about $\tau$ is vacuous.
So, assume that the statement is true for $i-1$. By the definition of the strategy, in round $i$,
for every $v\in V_i$,
\begin{eqnarray*}
\Pr[v_i = v] & = & \sum_{u\in V_{i-1}} \Pr[v_{i-1} = u] \cdot \Pr[v_i = v\mid v_{i-1} = u] \\
                  & = & \sum_{u\in V_{i-1}} P_{i-1}(u)\cdot \frac{\tau_i(u,v)}{P_{i-1}(u)} \\
                  & = & P_i(v),
\end{eqnarray*}
where the penultimate equality uses the induction hypothesis, and the final equality uses
the condition on the marginals of $\tau_i$ at $V_i$.
Similarly,
\begin{eqnarray*}
&    & \Pr[\hbox{the searcher moves from } u \hbox{ to } v] \\
& = & \Pr[v_{i-1} \!=\! u]\!\cdot\! \Pr[\hbox{the searcher moves from } u \hbox{ to } v\!\mid\! v_{i-1} \!=\! u] \\
& = & P_{i-1}(u)\cdot \frac{\tau_i(u,v)}{P_{i-1}(u)} \\
& = & \tau_i(u,v).
\end{eqnarray*}
Thus, by linearity of expectation, the searcher's expected total cost is
$$
\sum_{i=1}^n \sum_{u\in V_{i-1}} \sum_{v\in V_i} \tau_i(u,v)\cdot w_{G_i}(u,v),
$$
and this is by definition equal to the cost of the searcher's fractional strategy.
\end{proof}

\subsubsection{Layered trees}
\quad
\\

We now discuss special cases of layered graph traversal whose solution
implies a solution to the general case. We begin with a definition.
\begin{definition}
A rooted layered tree is an acyclic layered graph, where every vertex $v\ne a$ has exactly one 
neighbor in the preceding layer. We say that $a$ is the root of such a tree.
\end{definition}

\begin{theorem}[{Fiat et al.~\cite[Section 2]{FFKRRV91}}]\label{thm: layered trees suffice}
Suppose that the designer is restricted to play a width $k$ rooted layered tree with edge weights 
in $\{0,1\}$, and suppose that there is a $C$-competitive (pure or mixed) strategy of the searcher 
for this restricted game. Then, there is a $C$-competitive (pure or mixed, respectively) strategy 
of the searcher for the general case, where the designer can play any width $k$ layered graph 
with non-negative integer edge weights.
\end{theorem}

A width $k$ rooted layered tree is binary iff every vertex has at most two neighbors in the following 
layer. (Thus, the degree of each node is at most $3$.)
\begin{corollary}\label{cor: layered binary trees suffice}
The conclusion of Theorem~\ref{thm: layered trees suffice} holds if there is a $C$-competitive strategy 
of the searcher for the game restricted to the designer using width $k$ rooted layered binary trees with 
edge weights in $\{0,1\}$. Moreover, the conclusion holds if in addition we require that between two 
adjacent layers there is at most one edge of weight $1$.
\end{corollary}

\begin{proof}
Suppose that the designer plays an arbitrary width $k$ layered tree. The searcher converts the tree 
on-the-fly into a width $k$ layered binary tree, uses the strategy for binary trees, and maps the moves 
back to the input tree. The conversion is done as follows. Between every two layers that the designer 
generates, the searcher simulates $\lceil \log_2 k \rceil-1$ additional layers. If a vertex $u\in V_{i-1}$ 
has $m\le k$ neighbors $v_1,v_2,\dots,v_m\in V_i$, the searcher places in the simulated layers between 
$V_{i-1}$ and $V_i$ a layered binary tree rooted at $u$ with $v_1,v_2,\dots,v_m$ as its leaves. Notice 
that this can be done simultaneously for all such vertices in $V_{i-1}$ without violating the width constraint 
in the simulated layers. The lengths of the new edges are all $0$, except for the edges touching the leaves. 
For $j=1,2,\dots,m$, the edge touching $v_j$ inherits the length $w(\{u,v_j\})$ of the original edge $\{u,v_j\}$. 
Clearly, any path traversed in the simulated tree corresponds to a path traversed in the input tree that 
has the same cost---simply delete from the path in the simulated tree the vertices in the simulated layers; 
the edges leading to them all have weight $0$. The additional requirement is easily satisfied by now making
the following change. Between every two layers $i-1,i$ of the rooted layered binary tree insert $k-1$ simulated
layers. Replace the $j$-th edge between layer $i-1,i$ (edges are indexed arbitrarily) by a length $k$ path. If 
the original edge has weight $0$, all the edges in the path have weight $0$. If the original edge has weight 
$1$, then all the edges in the path have weight $0$, except for the $j$-th edge that has weight $1$.
\end{proof}

\subsection{Small set chasing}

This two-person game is defined with respect to an underlying metric space ${\cal M} = (X,\dist)$. 
The game alternates between the adversary and the algorithm. The latter starts at an arbitrary point
$x_0\in X$. The adversary decides on the number of rounds $n$ that the game will be played (this
choice is unknown to the algorithm until after round $n$). In round $i$ of the game, the adversary 
chooses a finite set $X_i\subset X$. The algorithm must then move to a point $x_i\in X_i$. The game 
is parametrized by an upper bound $k$ on $\max_{i=1}^n |X_i|$. The algorithm pays 
$\sum_{i=1}^n \dist(x_{i-1},x_i)$ and the adversary pays 
$$
\min\left\{\sum_{i=1}^n \dist(y_{i-1},y_i):\ y_0=x_0\wedge y_1\in X_1\wedge\cdots\wedge y_n\in X_n\right\}.
$$
\begin{theorem}[{Fiat et al.~\cite[Theorem 18]{FFKRRV91}}]
For every $k\in\NN$ and for every $C = C(k)$, there exists a pure (mixed, respectively) $C$-competitive 
online algorithm for cardinality $k$ set chasing in every metric space with integral distances iff there exists 
a pure (mixed, respectively) $C$-competitive online algorithm for width $k$ layered graph traversal.
\end{theorem}

\subsection{Reduction to evolving trees}

The following lemma is a more precise restatement of Theorem~\ref{thm:reductionShort} from the introduction. Combined with Theorem~\ref{thm:evolvingIntro}, it yields Theorem~\ref{thm:mainIntro}.
\begin{lemma}\label{lm: LGT to DTG reduction}
Let $k\in\NN$ and $C = C(k)$. 
Suppose that there exists a (pure, mixed, fractional, 
respectively) $C$-competitive strategy for the evolving tree game on binary trees of maximum depth 
$k$. 
Then, there exists a (pure, mixed, fractional, 
respectively) $C$-competitive strategy for traversing width $k$ layered graphs.
\end{lemma}

\begin{proof}
Consider at first fractional strategies.
By Lemma~\ref{lm: fractional to mixed} and Corollary~\ref{cor: layered binary trees suffice}, we can
restrict our attention to designing fractional strategies on width $k$ rooted layered binary trees with 
edge weights in $\{0,1\}$. Now, suppose that we have a fractional $C$-competitive strategy for the 
depth $k$ evolving tree game. We use it to construct a fractional $C$-competitive strategy for the
traversal of width $k$ rooted layered binary trees as follows. To simplify the proof, add a virtual layer 
$-1$ containing a single node $p(a)$ connected to the source $a$ with an edge of weight $0$.
We construct the layered graph strategy by induction over the current layer. Our induction hypothesis 
is that in the current state: 
\begin{enumerate}
\item The evolving tree is homeomorphic to the layered subtree spanned by the paths from $p(a)$ 
         to the nodes in the current layer. 
\item In this homeomorphism, $r$ is mapped to $p(a)$ and the leaves of the evolving tree are 
         mapped to the leaves of the layered subtree, which are the nodes in the current layer.
\item In this homeomorphism, each edge of the evolving tree is mapped to a path of the same 
         weight in the layered subtree. 
\item The probability assigned to a leaf by the fractional strategy for the evolving tree is equal to the 
         probability assigned to its homeomorphic image (a node in the current layer) by the fractional 
         strategy for the layered tree.
\end{enumerate}

Initially, the traversal algorithm occupies the source node $a$ with probability $1$. The evolving 
tree consists of the two initial nodes $r$ and $c_r$, with a $0$-weight edge connecting them. The 
homeomorphism maps $r$ to $p(a)$ and $c_r$ to $a$. The evolving tree algorithm occupies $c_r$ 
with probability $1$. Hence, the induction hypothesis is satisfied at the base of the induction. For 
the inductive step, consider the current layer $i-1$, the new layer $i$ and edges between them. If a node in layer $i-1$ has no child in layer $i$, we delete the homeomorphic 
preimage (which must be a leaf and cannot be $c_r$) in the evolving tree. If a node $v$ in layer $i-1$ has two children in layer $i$, we execute a fork step where we generate two new leaves and connect them to the preimage 
of $v$ (a leaf) in the evolving tree, and extend the homeomorphism to the new leaves in the obvious
way. Otherwise, if a node $v$ in layer $i-1$ has a single child in layer $i$, we modify the
homeomorphism to map the preimage of $v$ to its child in layer $i$. After executing as many such 
discrete steps as needed, if there is a weight $1$ edge connecting a node $u$ in layer $i-1$ to a 
node $v$ in layer $i$, we execute a continuous step, increasing the weight of the edge incident on
the homeomorphic preimage of $v$ in the evolving tree (which must be a leaf) for a time interval of length 
$1$. After executing all these steps, we simply copy the probability distribution on the leaves of the evolving 
tree to the homeomorphic images in layer $i$ of the layered tree. This clearly satisfies all the
induction hypotheses at layer $i$. Moreover, since the evolving tree has at most $k$ leaves at any time, and each non-leaf other than the root has exactly two children, its depth is at most $k$.

Notice that since the target $b$ is assumed to be the only node in the last layer, when we reach it, the 
evolving tree is reduced to a single edge connecting $r$ to the homeomorphic preimage of $b$. The weight 
of this edge equals the weight of the path in the layered tree from $p(a)$ to $b$, which is the same as the
weight of the path from $a$ to $b$ (because the edge $\{p(a),a\}$ has weight $0$). Moreover, the fractional
strategy that is induced in the layered tree does not pay more than the fractional strategy in the evolving
tree. Hence, it is $C$-competitive. 

Finally, deterministic strategies are fractional strategies restricted to probabilities in $\{0,1\}$, hence the claim 
for deterministic strategies is a corollary. This also implies the claim for mixed strategies, as they are probability 
distributions over pure strategies. 
\end{proof}

\section{Conclusion}
We have seen an $O(k^2\log d_{\max})$-competitive algorithm for the evolving tree game, and a resulting $O(k^2)$-competitive algorithm for width $k$ layered graph traversal and chasing sets of cardinality $k$, which is tight on account of the $\Omega(k^2)$ lower bound in~\cite{BCR22}. We note that the depth $k$ evolving tree game is strictly more general than width $k$ layered graph
traversal. In particular, the evolving binary tree corresponding to the width $k$ layered graph game 
has depth at most $k$ and also at most $k$ leaves. However, in general a depth $k$ binary tree may 
have $2^{k-1}$ leaves. Our evolving tree algorithm and analysis applies without further restriction on
the number of leaves.

The main conceptual message from our work is that the powerful framework of online mirror descent works even in evolving metric spaces. This may potentially have further applications for other problems. For example, obtaining a polylog$(k)$-competitive algorithm for the $k$-server problem seems to require embedding the underlying metric space into dynamically evolving hierarchically separated trees. In another interesting direction, a variant of the evolving tree game introduced here was recently used in~\cite{CossonM24} to obtain improved bounds for the collective tree exploration problem.

\section*{Acknowledgments}
We thank the anonymous reviewers at SICOMP for their constructive feedback that helped improve this article.

\bibliographystyle{plainurl}
\bibliography{biblio}

\end{document}

%% file: commands.tex
\newcommand{\dist}{\mathrm{dist}}

\renewcommand{\phi}{\varphi}

\newcommand{\R}{\mathbb{R}}

\newcommand{\cL}{\mathcal{L}}

\def\ds1{\mathds{1}}
\renewcommand{\epsilon}{\varepsilon}

\renewcommand{\tilde}{\widetilde}

\newlength{\minipagewidth}
\newcommand{\beq}{\begin{equation}}
\newcommand{\eeq}{\end{equation}}

\newcommand{\beqa}{\begin{eqnarray}}
\newcommand{\eeqa}{\end{eqnarray}}

\newcommand{\beqan}{\begin{eqnarray*}}
\newcommand{\eeqan}{\end{eqnarray*}}

\def\ba#1\ea{\begin{align*}#1\end{align*}} 
\def\banum#1\eanum{\begin{align}#1\end{align}} 